\documentclass[12pt]{article}

\usepackage{latexsym}
\usepackage{amsmath}
\usepackage{amssymb}
\usepackage{amsthm}
\usepackage{amscd}
\usepackage[mathscr]{eucal}
\usepackage{fullpage}
\usepackage{graphicx}
\usepackage{subfigure}
\usepackage{psfrag}
\usepackage{rotating}

\usepackage{color}

\newcommand{\R}{\mathbb{R}}

\definecolor{ggreen}{cmyk}{0.7,     0,      0.9,      0}
\definecolor{viol}{cmyk}{0.3,1,0,0}
\definecolor{myred}{cmyk}{0.1, 1, 0.5, 0}
\definecolor{bblue}{rgb}{0.2, 0.29996, 0.8 }

\theoremstyle{plain}

\newtheorem{theorem}{Theorem}[section]
\newtheorem{corollary}{Corollary}[section]

\newtheorem{proposition}{Proposition}[section]

\begin{document}


\title{\bf A stellar model with diffusion\\ in general relativity}
\author
{A.~Alho  \\
          {\small Center for Mathematical Analysis, Geometry and Dynamical Systems}  \\
     {\small Instituto Superior T\'ecnico, Universidade de Lisboa} \\
     {\small Av. Rovisco Pais, 1049-001 Lisboa, Portugal} \\
      {\small  \tt   aalho@math.ist.utl.pt}\\[0.4cm]
      S. Calogero  \\
      {\small Department of Mathematical Sciences}  \\
      {\small Chalmers University of Technology, University of Gothenburg}  \\
      {\small Gothenburg, Sweden} \\
      {\small  \tt  calogero@chalmers.se}
      }

\maketitle


\begin{abstract}
\noindent
We consider a spherically symmetric stellar model in general relativity whose interior consists of a pressureless fluid undergoing 
microscopic velocity diffusion in a cosmological scalar field. We show that the diffusion dynamics compel the interior to be spatially 
homogeneous, by which one can infer immediately that within our model, and in contrast to the diffusion-free case, no naked singularities 
can form in the gravitational collapse.  We then study the problem of matching an exterior Bondi type metric to the surface of the star 
and find that the exterior can be chosen to be a modified Vaidya metric with variable cosmological constant. Finally, we study in detail 
the causal structure of an explicit, self-similar solution.

\end{abstract}


\section{Introduction}
\label{sec:int}
The structure of singularities formed in the gravitational collapse of bounded matter distributions is a widely investigated problem in general 
relativity. The question of whether such singularities are naked, i.e., visible to far-away observers, or whether on the contrary they 
are safely hidden inside a black hole, has been the subject of innumerable works in the physical and mathematical literature. 
Nevertheless the problem remains poorly understood, even in the idealized case in which the collapsing body is spherically symmetric.
A complete solution is only available for the gravitational collapse of a spherically symmetric dust cloud and can be found in the pioneering works by 
~Oppenheimer-Snyder~\cite{OS39} and Christodoulou~\cite{Chr84}.  (In~\cite{Chr94,Chr99} Christodoulou analyses the question of existence and stability of naked singularities in the 
gravitational collapse of a massless scalar field.  Notwithstanding the importance of these works, our focus is on matter models that describe 
material bodies, such as perfect fluids and kinetic particles~\cite{And}.) 

The Oppenheimer-Snyder model consists of a collapsing spatially homogeneous and isotropic dust interior, described by the contracting 
Friedmann-Lema\^itre solution, matched at a comoving boundary with a Schwarzschild exterior. Dropping the 
homogeneity assumption of the interior leads to the class of Lema\^itre-Tolman-Bondi solutions. These inhomogeneous stellar 
models were studied numerically by Eardley and Smarr~\cite{ES79}, and analytically by Christodoulou~\cite{Chr84}. 
It was shown that, in contrast to the case studied by Oppenheimer and Snyder, the spatially inhomogeneous collapse leads to the formation of naked singularities. See the prologue of~\cite{Chr08} for an historical review on the gravitational collapse problem in general relativity. 

In this paper we initiate the study of the gravitational collapse of matter subject to diffusion. We believe that the inclusion of 
diffusion dynamics in the gravitational collapse problem is meaningful both from a mathematical and physical point of view. 
From one hand it is well known that diffusion terms introduce a regularizing effect in the equations, which might prevent the 
formation of naked singularities in general relativity. 
On the other hand the physical relevance of diffusion phenomena is unquestionable and the applications in general relativity have 
been discussed in~\cite{calogero1,calogero2,calogero3}.

We begin our study with the simplest possible model, namely a spherically symmetric dust cloud undergoing diffusion in a cosmological 
scalar field. In this case the regularizing effect due to diffusion is overwhelming: The interior of the dust cloud is forced to be 
spatially homogeneous. By this fact one can easily infer that, in contrast to the diffusion-free scenario described above, naked 
singularities cannot form in the gravitational collapse of a spherically symmetric dust cloud in the presence of diffusion.


Another interesting property of our model is that, in contrast to the diffusion-free case~\cite{OS39,FST96,Chr84}, 
the exterior of the star cannot be static. The simplest spherically symmetric solution of the Einstein equations that can provide 
a suitable exterior region for our stellar model is given by a Vaidya type metric which includes a variable 
cosmological constant (a generalization of the radiating version of the Schwarzschild-(Anti-)de-Sitter family of solutions). 

A detailed analysis of our model is given in the following sections. We conclude this introduction by outlining the diffusion theory of matter in general relativity. 
There exist two versions of this theory: a kinetic one~\cite{calogero1}, which is based on a Fokker-Planck equation for the particle density in phase-space, 
and a fluid one~\cite{calogero2}, which is the (formal) macroscopic limit of the kinetic theory.
In the present paper we apply the fluid theory. We recall that the energy-momentum tensor and current density of a perfect fluid are
\begin{equation}
T^{\mu\nu}=\rho u^\mu u^\nu +p(g^{\mu\nu}+u^\mu u^\nu ),  \qquad
J^\mu =n u^\mu,
\label{TJ}
\end{equation}
where $\rho$ is the rest-frame energy density, $p$ the pressure, $u$ the four-velocity and $n$ the
particle density of the fluid. The diffusion behavior is imposed by postulating the equations
\begin{subequations}\label{mattereq}
\begin{align}
&\nabla_\mu T^{\mu\nu}=\sigma J^\nu,\label{diffeq}\\
& \nabla_\mu(n u^\mu )=0,\label{consn}
\end{align}
\end{subequations}
where $\sigma>0$ is the diffusion constant, which measures the energy gained by the particles  per unit
of time due to the action of the diffusion forces. The second equation entails the conservation of the total number of fluid particles.

By projecting~\eqref{diffeq} along the direction of $u^\mu$ and onto the hypersurface orthogonal to $u^\mu$, we obtain the following equations on the matter fields:
\begin{subequations}\label{mattereq2}
\begin{eqnarray}
& \nabla_{\mu} (\rho u^{\mu})+p\nabla_{\mu} u^{\mu}=\sigma n,\label{ProjTJ1}\\
&(\rho +p) u^\mu \nabla_{\mu}u^\nu +u^\nu u^\mu \nabla_{\mu}p+\nabla^{\nu} p=0.\label{ProjTJ2}
\end{eqnarray}
\end{subequations}
The system~\eqref{mattereq2} on the matter variables must be completed by assigning an equation of state between the pressure, the energy density and the particles number density. In this paper we assume that the fluid is pressureless (dust fluid).
As the energy-momentum tensor of the fluid is not divergence-free, 
see~\eqref{diffeq}, we have to postulate the existence of an additional matter field in spacetime to re-establish the (local) conservation 
of energy. The role of this additional matter field is that of the solvent matter in which the diffusion of the fluid particles takes 
place. The simplest choice is to assume the existence of a vacuum energy scalar field $\phi$ with energy-momentum tensor 
$-\phi g_{\mu\nu}$, which leads to the following Einstein equations for the spacetime metric $g$ (in units $8\pi G=c=1$):
\begin{equation}\label{EinsteinGEN}
R_{\mu\nu}-\frac{1}{2}g_{\mu\nu}R+\phi g_{\mu\nu}=T_{\mu\nu}.
\end{equation}
The evolution equation on the scalar field $\phi$ determined by~\eqref{EinsteinGEN}, the Bianchi identities, 
and the diffusion equation~\eqref{diffeq} is
\begin{equation}\label{phieq}
\nabla_\nu \phi =\sigma J_\nu.
\end{equation}
%


\section{The stellar model}
Throughout the rest of the paper we assume that spacetime $(M,g)$ is spherically symmetric. The particle number density  $n:M\to [0,\infty)$ and the four-velocity $u_x:T_xM\to\R^4$ at each point $x\in M$ are also spherically symmetric and satisfy~\eqref{consn}. Under these assumptions  one can cover an open neighborhood $U$ of the center of symmetry by a coordinate system $(t,R,\theta,\psi)$ such that the metric and the four-velocity take the form
\[
g=-e^{2\Phi}\,dt^2+e^{2\Psi}\,dR^2+r^2\,d\Omega^2\quad u=e^{-\Phi}\partial_t\quad \text{in $U$}.
\]
Here $\Phi,\Psi,r$ are functions of $(t,R)$ and $d\Omega^2=d\theta^2+\sin^2\theta d\psi^2$ is the standard metric on $S^2$. The center of symmetry is defined by the timelike curve $r(t,R)=0$. This coordinate system is called {\it comoving} and it is defined up to a transformation $t\to F(t), R\to G(R)$ of the time and radial coordinates. 

\subsection{The interior}
We assume that the interior of the star defines a region $V$ of spacetime covered by comoving coordinates. In particular, we assume that there exist $T\in (0,+\infty]$ and $R_b>0$ such that $0\leq t<T$ and $0\leq R<R_b$ within $V\subset U$. The timelike hypersurface given by
\[
\Sigma\ :\ R=R_b,
\] 
will be identified with the surface of the star. 
To identify the region $V$ as the interior, we assume that the matter fields $(\rho,n)$ are nowhere vanishing on $V$.  
The freedom in the choice of the comoving coordinates will be used to impose
\begin{equation}\label{bcPhi}
\Phi(t,R_b)=0,
\end{equation}
so that $t$ is the proper time of observers at rest with respect to the boundary of the star, and
\begin{equation}\label{bcr}
r(0,R)=R,
\end{equation}
so that the comoving radius $R$ coincides initially, i.e., at time $t=0$, with the radius function of the group orbits. 
It will now be proved that, when the fluid is pressureless, the diffusion dynamics compel the interior of our model to be spatially homogeneous. We denote $\dot f=\partial_t f$, $f'=\partial_R f$, for any function $f=f(t,R)$, and similarly $\dot y$ denotes the 
derivative of any function $y$ of one variable.

\begin{theorem}
Let $p=0$ and let $(g,\rho,n,u,\phi)$ be a spherically symmetric solution of~\eqref{mattereq2}-\eqref{phieq} on $V$. Then $\rho$, $n$, $\phi$ are functions of $t\in[0,T)$ only and there exist a positive function $a:[0,T)\to (0,\infty)$ and a constant $k\in\R$ such that
\begin{equation}
g=-dt^2+a(t)^2(\frac{dR^2}{1-kR^2}+R^2d\Omega^2),\quad u=\partial_t.
\end{equation}
\end{theorem}
\begin{proof}
Equation~\eqref{ProjTJ2} for $p=0$ reads $\rho \Phi'=0$, which together with the boundary condition~\eqref{bcPhi} implies
$\Phi=0$. Hence $u=\partial_t$, as claimed.
Equation~\eqref{consn} now reads
\begin{equation}\label{dJ=0com}
\partial_t( nr^2 e^{\Psi})=0.
\end{equation}
Moreover~\eqref{phieq} implies that $\phi$ and $n$ are functions of $t$ only, related by
\begin{equation}\label{phisol2}
\phi(t)=\phi(0)-\sigma\int_0^t n(s)\,ds.
\end{equation}
The Einstein equations~\eqref{EinsteinGEN} are:
\begin{align}
&-2\frac{r''}{r}+2\frac{r'}{r}\Psi'-\left(\frac{r'}{r}\right)^2+e^{2\Psi}\left(\left(\frac{\dot{r}}{r}\right)^2+2\frac{\dot r}{r}\dot\Psi\right)+\frac{e^{2\Psi}}{r^2}=(\rho+\phi)e^{2\Psi},\label{ham}\\
&e^{-2\Psi}\left(\frac{r'}{r}\right)^2-\frac{1}{r^2}-\left(\frac{\dot r}{r}\right)^2-2 \frac{\ddot r}{r}+\phi=0,\label{second}\\
& -\ddot\Psi-\dot\Psi^2-\frac{\ddot r}{r}-\frac{\dot r}{r}\dot\Psi+e^{-2\Psi}\left(\frac{r''}{r}-\frac{r'}{r}\Psi'\right)+\phi=0,\\
& r'\dot\Psi-\dot r'=0,
\end{align}
where the only non-zero component of the energy momentum tensor is $T_{00}=\rho$. The last Einstein equation gives
\begin{equation}\label{psi}
e^{2\Psi}=\frac{{r'}^2}{f(R)^2},
\end{equation}
where $f(R)$ is an arbitrary positive function of the radial variable. Substituting into~\eqref{dJ=0com} we obtain
\[
\partial_t(n r^2 r')=0.
\]
Integrating with the boundary condition~\eqref{bcr} we obtain
\begin{equation}\label{K}
r(t,R)=R\left(\frac{n(0)}{n(t)}\right)^{1/3}:=R a(t),
\end{equation}
where we denoted $a(t)=(n(0)/n(t))^{1/3}$.
Substituting~\eqref{psi} and \eqref{K} into~\eqref{second} we obtain the equation
\[
[2 a(t)\ddot{a}(t)+\dot{a}(t)^2-\phi(t) a(t)^2]R^2+1-f^2(R)=0.
\]
The latter implies that there exists a constant $k\in\R$ such that
\[
1-f^2(R)=kR^2,\quad 2\frac{\ddot{a}}{a}+\left(\frac{\dot{a}}{a}\right)^2-\phi=-\frac{k}{a^2}.
\]
Finally~\eqref{ham} entails that $\rho$ is a function of $t$ only.
\end{proof}
From now on we assume that the interior fluid is pressureless.
By rescaling the radial coordinate we may assume that $k\in\{-1,0,1\}$, hence the spacetime metric in the interior of the star takes 
the Robertson-Walker form
\begin{equation}\label{interiormetric}
g_\mathrm{int}=-dt^2+a(t)^2(\frac{dR^2}{1-kR^2}+R^2d\Omega^2),\quad k=0,\pm1,\quad R<R_b,
\end{equation}
where $R_b<1$ for $k=1$, and the equations ~\eqref{mattereq2}-\eqref{phieq} reduce to the following. The conservation of the particle 
number density~\eqref{consn} gives $n(t)=a(0)^3 n(0)/a(t)^3$, while $a(t),\rho(t),\phi(t)$ satisfy
\begin{subequations}\label{einsteineqs}
\begin{align}
&\Big(\frac{\dot a}{a}\Big)^2+\frac{k}{a^2}=\frac{1}{3}(\rho+\phi),\label{const}\\
&\frac{\ddot a}{a}=-\frac{1}{6}\rho+\frac{1}{3}\phi,\label{friedmann}\\
&\dot\phi=-3\frac{\beta}{a^3},\label{phieq2}\\
&\dot\rho=-3\rho\frac{\dot a}{a}-\dot\phi,\label{rhoeq}
\end{align}
\end{subequations}
where 
\begin{equation}\label{beta}
\beta=\sigma n(0) a(0)^3/3. 
\end{equation}
The initial data for the system~\eqref{einsteineqs} are given on $\{t=0\}\cap V$ and consist of a 
quadruple $(a_0,\dot a_0,\rho_0,\phi_0)$, with $a_0,\rho_0>0$, such that the constraint equation~\eqref{const} is satisfied at time $t=0$, 
i.e.,
\[
\Big(\frac{\dot a_0}{a_0}\Big)^2+\frac{k}{a_0^2}=\frac{1}{3}(\rho_0+\phi_0).
\] 
In particular, the set of admissible initial data comprises a three-dimensional manifold, which we denote by $\mathcal{I}$.  
By a regular solution of~\eqref{einsteineqs} in the interval $[0,T)$ with  initial data $(a_0,\dot a_0,\rho_0,\phi_0)\in\mathcal{I}$ 
we mean a triple of functions $a\in C^2((0,T))$, $\phi,\rho\in C^1((0,T))$ satisfying~\eqref{friedmann}-\eqref{rhoeq} and such that 
$a(t),\rho(t)>0$, for $t\in [0,T)$, and $\lim_{t\to 0^+}(a(t),\dot a(t),\rho(t),\phi(t))=(a_0,\dot a_0,\rho_0,\phi_0)$. 
Let $T_\mathrm{max}$ be the maximal time of existence of a regular solution with a given set of initial data. 
We say that the regular solution is global if $T_\mathrm{max}=+\infty$.

For $\beta=0$, i.e., in the absence of diffusion, the interior reduces to a dust cloud in a spacetime with cosmological constant 
$\Lambda=\phi_0$. In particular, for $\beta=0$ and $\phi_0=0$, the exterior must be Schwarzschild, and we recover the well-known 
Oppenheimer-Snyder model~\cite{OS39}. The analysis of the latter model is greatly simplified by the fact that a closed formula solution 
is known for the Einstein equations, see eg.~\cite{GP09}. In contrast to this, the general solution to the system~\eqref{einsteineqs} 
is not known, and so the analysis of the stellar interior in the presence of diffusion is more complicated. 
This analysis has been carried out in~\cite{AC15}, where a complete characterization of the qualitative behavior of solutions to the system~\eqref{einsteineqs} has been obtained using dynamical systems methods. The relevant properties for the present study are summarized in the following theorem:
\begin{theorem}\label{theo}
The manifold $\mathcal{I}$ of initial data can be written as the disjoint union of two three-dimensional submanifolds $\mathcal{I}_\mathrm{exp}$, $\mathcal{I}_{col}$ such that
\begin{itemize}
\item[(i)] For initial data in the submanifold $\mathcal{I}_\mathrm{exp}$, the corresponding regular solution of~\eqref{einsteineqs} 
is global and $\dot{a}(t)>0$ 
holds for all $t\geq 0$. We call these interior models ``expanding".
\item[(ii)] For initial data in the submanifold $\mathcal{I}_\mathrm{col}$, the corresponding regular solution of~\eqref{einsteineqs} 
exists only up to a finite time $T_\mathrm{max}=t_s>0$ and $a(t)\to 0$ as $t\to t_s^-$. 
Moreover 
there exists a constant $c>0$ such that
\[
a(t)\sim c(t-t_s)^{2/3},\quad \dot{a}(t)\sim \frac{2c}{3} (t-t_s)^{-1/3}, \quad \ddot{a}(t)\sim -\frac{2c}{3}(t-t_s)^{-5/3}, \quad\text{as $t\to t_s^-$}.  
\]
We call these interior models ``collapsing".
\end{itemize} 
\end{theorem}
The asymptotic behavior of the scalar factor claimed in (ii) corresponds to the property proved in~\cite{AC15} that collapsing (dust) 
solutions behave like the Friedmann-Lema\^{i}tre diffusion-free solution in the limit toward the singularity. It implies that the spacelike singularity at $t=t_s$ is a 
curvature singularity, as the Kretschmann scalar $K=\mathrm{Riem}^2$ satisfies
\[
K=O((t-t_s)^{-4}),\quad \text{as $t\to t_s^-$.}
\]
In particular, spacetime is inextendible beyond the spacelike hypersurface $t=t_s$ and no outgoing light ray can emanate from the 
singularity. We obtain the following important corollary:
\begin{corollary}
There exist no naked singularities in the spherical collapse of dust clouds undergoing diffusion in a cosmological scalar field.
\end{corollary}
We remark that the above result is independent of the spacetime exterior and applies to all collapsing models. 
In the absence of diffusion the formation of (locally and globally) naked singularities in the gravitational collapse of a dust cloud 
occurs for spatially inhomogeneous models, see~\cite{Chr84}. In the present case the action of the diffusion forces prevents the formation 
of naked singularities by an explicit and compelling regularizing effect: it forces the dust interior to be spatially homogeneous. 
 
Let us recall the definition of some important geometrical/physical quantities associated to the metric~\eqref{interiormetric}. Let
\begin{equation}\label{chi}
\chi(t,R)=1-kR^2-\dot a(t)^2R^2=(\sqrt{1-kR^2}+\dot{a}(t)R)(\sqrt{1-kR^2}-\dot{a}(t)R).
\end{equation}
The region in the interior where $\chi<0$ is the region of trapped surfaces, while $\chi>0$ defines the regular interior region. 
The hypersurface $\chi=0$ defines the apparent horizon.
We define a local mass function $m(t,R)$ through the identity
\begin{equation}\label{mass}
\chi=1-\frac{2m}{r}-\frac{\phi}{3}\,r^2,\quad r(t,R)=a(t)R,
\end{equation}
which reduces to the standard definition of Misner-Sharp mass when $\phi\equiv 0$~\cite{FST96}. By~\eqref{const} we obtain the important identity
\[
m(t,R)=a(t)^3R^3\frac{\rho(t)}{6}.
\]
Moreover by~\eqref{phieq2}-\eqref{rhoeq} we have
\[
\rho(t)a(t)^3=\rho_0a_0^3+3\beta t,
\]
hence the mass function is given by
\begin{equation}\label{dtmassdust}
m(t,R)=m(0,R)+\frac{R^3}{2}\beta t.
\end{equation}
Thus the local mass is conserved only in the absence of diffusion and it is otherwise linearly increasing in time. The latter behavior is intimately connected with the irreversibility of the diffusion process. In fact, using that the entropy $S$ of a dust fluid equals the energy per particle, i.e., $S=\rho/n$ (see~\cite{calogero2}), and due to the conservation of the total number of particles, i.e., $n(t)a(t)^3=n(0)a(0)^3$, we have
\[
S(t)=\frac{\rho(t)}{n(t)}=\frac{\rho(t)a(t)^3}{n(0)a(0)^3}=\frac{S(0)}{m(0,R)}m(t,R),
\]
hence by~\eqref{dtmassdust} and recalling the definition~\eqref{beta} of the parameter $\beta$ we obtain that the entropy is linearly increasing in time:
\[
S(t)=S(0)+\sigma t.
\]
\subsection{The exterior}
We write the metric on the exterior $V^c$ using Bondi coordinates:
\begin{equation}\label{exteriormetric}
g_\mathrm{ext}=-A B dw^2+2\varepsilon A dwdr+r^2d\Omega^2,
\end{equation}
where $A,B$ are functions of $(w,r)$ and $\varepsilon=\pm 1$. The time coordinate $w$ is the ingoing (advanced) null coordinate for 
$\varepsilon=1$ and the outgoing (retarded) null coordinate for $\varepsilon=-1$. We assume that the comoving boundary $\Sigma$ between the interior and the exterior region is expressed as $r=r_\Sigma(w)$ in Bondi coordinates, i.e.,
 \[
 \Sigma: r=r_\Sigma(w),
 \]
 so that $r>r_\Sigma(w)$ in~\eqref{exteriormetric}. As the boundary $\Sigma$ is assumed to be timelike, we require that
 \begin{equation}\label{timelikecond}
 A(w,r(w))(B(w,r(w))-2\varepsilon\dot{r}_\Sigma(w))>0,\quad \dot{r}_\Sigma=\frac{dr_\Sigma}{dw},
 \end{equation}
i.e., the first fundamental form of $\Sigma$ has the signature $(-,+,+)$. 
Moreover, letting $t=t(w)$ be the transformation of the time coordinate on the boundary, we require $dt/dw>0$, which means that the time 
orientation of spacetime does not change across the boundary.
 We recall that two metrics $g_\mathrm{int}$ and $g_\mathrm{ext}$ may be matched on $\Sigma$ if and only if they satisfy the junction 
 conditions that they induce the same first and second fundamental form on $\Sigma$, see~\cite{FST96} and references therein.

\begin{theorem}\label{jctheo}
The metrics~\eqref{interiormetric} and~\eqref{exteriormetric} satisfy the junction conditions on the comoving boundary $\Sigma$ if and only if
\begin{itemize}
\item[(a)] There holds
\begin{equation}\label{rsigma}
r_\Sigma(w)=a(t(w))R_b;\\ 
\end{equation}
\item[(b)] The transformation of variable $t=t(w)$ satisfies
\begin{equation}\label{dtdw2}
\dot{t}(w)=A(w,r_\Sigma(w))(\sqrt{1-k R_b^2}-\varepsilon\dot{a}(t(w))R_b);
\end{equation}
\item[(c)] There holds
\begin{equation}\label{B}
B(w,r_\Sigma(w))=A(w,r_\Sigma(w))(1-kR_b^2-\dot a(t(w))^2R_b^2);
\end{equation}
\item[(d)] There holds $Q(w)=0$, where
\begin{align}
Q(w)=&\Bigg[(B-2\varepsilon \dot{r}_\Sigma(w))((B-\varepsilon\dot{r}_\Sigma(w)) \partial_rA+\varepsilon\partial_wA)\nonumber\\
& \quad-A(2\ddot{r}_\Sigma(w)+(B-3\varepsilon\dot{r}_\Sigma(w))\partial_rB-\varepsilon\partial_wB) \Bigg]_{r=r_\Sigma(w)}.\label{Q}
\end{align}
\end{itemize}
\end{theorem}
\begin{proof}
The first and second fundamental forms induced by the interior metric~\eqref{interiormetric} on $\Sigma$ are given respectively by 
\[
\bar{g}_\mathrm{int}=-dt^2+a(t)^2R_b^2d\Omega^2,\qquad
\bar{K}_\mathrm{int}=R_b\sqrt{1-kR_b^2}\,a(t)d\Omega^2,
\]
while the same quantities induced by the exterior metric~\eqref{exteriormetric} are 
\[
\bar{g}_\mathrm{ext}=-A(w,r_\Sigma(w))\Big[B(w,r_\Sigma(w))-2\varepsilon \dot{r}_\Sigma(w)\Big]dw^2+r_\Sigma(w)^2d\Omega^2,
\]
\begin{align*}
\bar{K}_\mathrm{ext}&=-\frac{Q(w)}{2[A(w,r_\Sigma(w))(B(w,r_\Sigma(w))-2\varepsilon\dot{r}_\Sigma(w))]^{1/2}} dw^2\\[0.5mm]
&\quad +\frac{r_\Sigma(w)(B(w,r_\Sigma(w))-\varepsilon \dot{r}_\Sigma(w))}{[A(w,r_\Sigma(w))(B(w,r_\Sigma(w))-2\varepsilon  \dot{r}_\Sigma(w))]^{1/2}}d\Omega^2.
\end{align*}
Hence the junction condition $\bar{g}_\mathrm{int}=\bar{g}_\mathrm{ext}$ gives immediately~\eqref{rsigma} as well as
\begin{equation}
\dot{t}(w)=\Big[ A(w,r_\Sigma(w))\Big(B(w,r_\Sigma(w))-2\varepsilon \dot{r}_\Sigma(w)\Big)\Big]^{1/2}\label{dtdw}.
\end{equation}
The junction condition $\bar{K}_\mathrm{int}=\bar{K}_\mathrm{ext}$ is equivalent to $Q(w)=0$ and
\[
\frac{r_\Sigma(w)(B(w,r_\Sigma(w))-\varepsilon \dot{r}_\Sigma(w))}{[A(w,r_\Sigma(w))(B(w,r_\Sigma(w))-2\varepsilon \dot{r}_\Sigma(w))]^{1/2}}=R_ba(t(w))\sqrt{1-kR_b^2}.
\]
Using~\eqref{rsigma} in the latter equation we obtain
\[
B(w,r_\Sigma(w))[B(w,r_\Sigma(w))-2\varepsilon\dot{r}_\Sigma(w)]+\dot{r}_\Sigma(w)^2=(1-kR_b^2)A(w,r_\Sigma(w))[B(w,r_\Sigma(w))-2\varepsilon\dot{r}_\Sigma(w)].
\]
In the term $\dot{r}_\Sigma(w)^2$ we use
\[
\dot{r}_\Sigma(w)=\frac{dr_\Sigma}{dt}\frac{dt}{dw}=\dot a(t(w))R_b\Big[ A(w,r_\Sigma(w))\Big(B(w,r_\Sigma(w))-2\varepsilon \dot{r}_\Sigma(w)\Big)\Big]^{1/2}
\]
and so doing we obtain~\eqref{B}. Replacing~\eqref{rsigma} and~\eqref{B} into~\eqref{dtdw} gives~\eqref{dtdw2}. 
\end{proof}
In principle, any exterior metric that satisfies the conditions in the theorem can be matched to the interior. However a particularly 
simple and natural choice can be made as follows. We assume that $A(w,r)=1$ in the whole exterior, so that~\eqref{B} gives 
\begin{equation}\label{temporalino}
B(w,r_\Sigma(w))=1-kR_b^2-\dot{a}(t(w))^2R_b^2.
\end{equation}
Using~\eqref{chi} and~\eqref{mass} in~\eqref{temporalino} we obtain
\begin{equation}\label{B2}
B(w,r_\Sigma(w))=1-\frac{2m(t(w),R_b)}{r_\Sigma(w)}-\frac{\phi(t(w))}{3}r_\Sigma(w)^2.
\end{equation}
Moreover the junction condition (d) in Theorem~\ref{jctheo} becomes
\[
\Bigg[2\ddot{r}_\Sigma(w)+(B-3\varepsilon\dot{r}_\Sigma(w))\partial_rB-\varepsilon\partial_wB \Bigg]_{r=r_\Sigma(w)}=0.
\]
Using~\eqref{friedmann} and~\eqref{B2} we may write the latter as
\begin{equation}\label{tempino}
\Bigg[\frac{2}{r}-\frac{6m(t(w),R_b)}{r^2}-\frac{2B}{r}+\partial_rB\Bigg]_{r=r_\Sigma(w)}=0.
\end{equation}
If we now require~\eqref{tempino} to be valid for {\it all} $r>r_\Sigma(w)$, and not only on $\Sigma$, and solve the resulting differential equation on $B$ subject to the boundary condition~\eqref{B2} at $r=r_\Sigma(w)$ we obtain
\[
B(w,r)=1-\frac{2m(t(w),R_b)}{r}-\frac{\Lambda(w)}{3}\,r^2,
\]
where
\begin{equation}\label{cosmoconst}
\Lambda(w)=\phi(t(w)).
\end{equation}
Hence we obtain the following corollary.
\begin{corollary}
The exterior metric can be chosen to be the Vaidya type metric with variable cosmological constant given by
\begin{equation}\label{vaidya}
g_\mathrm{ext}=-\big(1-\frac{2M(w)}{r}-\frac{\Lambda(w)}{3}\,r^2\big) dw^2+2\varepsilon dwdr +r^2d\Omega^2,
\end{equation}
where $\Lambda(w)$ is given by~\eqref{cosmoconst} and 
\begin{equation}\label{massexterior}
M(w)=m(t(w),R_b)
\end{equation}
is the exterior mass function.
\end{corollary}
Conversely, if we now start with a metric of the form~\eqref{vaidya} and impose the matching conditions~\eqref{B2}-\eqref{tempino}, we obtain the identities $M(w)=m(t(w),R_b)$ and $\Lambda(w)=\phi(t(w))$, that is to say, the mass function and the cosmological scalar field must be continuous through the boundary in order for the metric~\eqref{vaidya} to be an admissible exterior.

The generalized Vaidya metric~\eqref{vaidya} solves the Einstein equation~\eqref{EinsteinGEN}
with cosmological scalar field $\phi=\Lambda(w)$ and the energy-momentum tensor
\[
T_\mathrm{ext}=\tilde{\rho}\,dw^2,\quad \tilde{\rho}=\frac{\varepsilon}{r^2}\Big(2\frac{dM}{dw}+\frac{r^3}{3}\frac{d\Lambda}{dw}\Big),
\]
where, by~\eqref{phieq2},~\eqref{dtmassdust} and~\eqref{dtdw2},
\begin{align}
&\frac{dM}{dw}=\frac{\beta R_b^3}{2}(\sqrt{1-kR_b^2}-\varepsilon\dot{a}(t(w))R_b),\label{dmdw}\\[0.3cm]
&\frac{d\Lambda}{dw}=-\frac{3\beta}{a(t(w))^3}(\sqrt{1-kR_b^2}-\varepsilon\dot{a}(t(w))R_b).\label{dldw}
\end{align}
When $\beta=0$, i.e., in the absence of diffusion, the functions $M$ and $\Lambda$ become constant and so $T_\mathrm{ext}=0$. 
In this limit our model reduces to the Oppenheimer-Snyder model with cosmological constant~\cite{nakao}.
The energy-momentum tensor $T_\mathrm{ext}$ describes a cloud of dust particles moving along the null directions orthogonal to $w=const$. 
Observe that the dust particles in the exterior {\it do not } undergo diffusion in the cosmological scalar field. 
This outcome of the model is consistent with the well-known fact that particles moving along null directions are not 
subject to diffusion, see~\cite{dudley}. Requiring the weak energy condition $\tilde{\rho}>0$ to hold in the whole exterior region 
forces us to restrict to the outgoing Vaidya solution.
\begin{proposition}\label{weakencon}
$\tilde{\rho}(w,r_\Sigma(w))=0$. Moreover the weak energy condition $\tilde{\rho}(w,r)>0$, for all $r>r_\Sigma(w)$, holds only for the outgoing Vaidya metric~\eqref{vaidya}$_{\varepsilon=-1}$.
\end{proposition}
\begin{proof}
By~\eqref{dmdw} and~\eqref{dldw},
\begin{equation}\label{tac}
\tilde{\rho}(w,r)=\frac{\varepsilon\beta\frac{dt}{dw}}{r^2a(t(w))^3}\big[r^{3}_\Sigma(w)-r^3\big],\quad r\geq r_\Sigma(w),
\end{equation}
by which the result follows immediately.
\end{proof}
We remark that the property $\tilde{\rho}(w,r_\Sigma(w))=0$ implies that the surface of the star is not radiating, i.e., the star is thermally isolated from the exterior. In particular  the null dust particles in the exterior are not emanated by the star, but rather they  are created spontaneously by the 
decaying vacuum energy field $\Lambda$ in the exterior (see~\cite{calogero2} for a similar behavior occurring in cosmology). 
Note also that for $\varepsilon=+1$ we have $\tilde{\rho}>0$ for $r<r_{\Sigma}(w)$. Hence if we consider 
a spatially homogeneous exterior, the interior can be chosen to be the ingoing Vaidya metric, leading to a void model with 
diffusion which generalises the well-known Einstein-Strauss model~\cite{ES45}. 

It is worth to point out the following important differences between our model and the analogous diffusion-free model studied 
in~\cite{FJLS,BOS89}. Firstly, in the diffusion-free model the Vaidya solution can arise as the exterior of a fluid ball only if the 
pressure of the fluid is positive (dust is not allowed). Secondly the matching boundary in the diffusion-free model cannot be 
comoving. The latter property can be naturally regarded as being the cause of energy dissipation from the surface of the star 
as seen by comoving observers, see the discussion in~\cite{FJLS,BOS89}. 

In view of Proposition~\ref{weakencon} we restrict from now on to the generalized outgoing Vaidya metric~\eqref{vaidya}$_{\varepsilon=-1}$, which we rewrite in standard notation as
\begin{equation}\label{vaidya2}
g_\mathrm{ext}=-\big(1-\frac{2M(u)}{r}-\frac{\Lambda(u)}{3}r^2\big) du^2-2 dudr +r^2d\Omega^2.
\end{equation}
The apparent horizons of the metric~\eqref{vaidya2} are the hypersurfaces in the exterior where 
\[
B(u,r)=\big(1-\frac{2M(u)}{r}-\frac{\Lambda(u)}{3}r^2\big)=0.
\]
As $B(u,r_\Sigma(u))=1-kR_b^2-\dot a (t(u))^2R_b^2$, see~\eqref{B}, and owing to~\eqref{chi}, the interior and exterior apparent horizons intersect on the boundary. 

\section{Self-similar solutions}
In this section we study in detail the causal structure of a particular solution to our model. Recall that the interior is assumed to be pressureless. Assuming in addition that it is self-similar, we obtain the following solution of~\eqref{EinsteinGEN}:
\begin{subequations}\label{sss}
\begin{align}
&a(t)=\delta_k t,\\
&\phi(t)=\frac{3\beta}{2\delta_k}a(t)^{-2},\\
&\rho(t)=\frac{3\beta}{\delta_k}a(t)^{-2},
\end{align}
where $\delta_k$ is the real (positive) solution of the polynomial equation
\begin{equation}\label{poleq}
\delta^3+k\delta-\frac{3\beta}{2}=0.
\end{equation} 
\end{subequations}
The Penrose diagrams of this solution are given in Figure~\ref{fig1}.
It has been shown in~\cite{AC15} that, for an open set of initial data, solutions of both types 
$\mathcal{I}_\mathrm{exp}$ and $\mathcal{I}_\mathrm{col}$ have an intermediate behavior that is close to the special solution~\eqref{sss}.

\begin{figure}[ht!]
\begin{center}
\psfrag{i--}[cc][cc][0.8][0]{$i^{-}$}
\psfrag{i-}[cc][cc][0.8][0]{$i^{-}\,/\,\mathscr{I^-}(t=0)$}
\psfrag{i0}[cc][cc][0.8][0]{$i^{0}$}
\psfrag{i++}[ll][cc][0.8][0]{$i^{+}$}
\psfrag{i+}[cc][cc][0.8][0]{$i^{+}\,/\,\mathscr{I^+}(t=\infty)$}
\psfrag{R}[rc][cc][0.8][00]{$R=0$}
\psfrag{H}[ll][cc][0.8][90]{$R=R_{AH}$}
\psfrag{t}[ll][bb][0.8][0]{$\mathscr{I^-}(t=0)$}
\psfrag{J}[ll][bb][0.8][0]{$\mathscr{I^+}(t=\infty)$}
\subfigure[Conformal diagram and bounded conformal diagram for $k=1$.
The solid lines correspond to the boundary $R=0$  
and the dashdotted line to the equator $R=1$. The remaining dotted lines are curves of constant $R$, for $0<R<1$,
while the dashed lines represent the apparent horizon at $R=R_{AH}$.
In this case a suitable matching surface is given by the curve $R=R_b<1$.
]
{\includegraphics[width=0.45\textwidth]{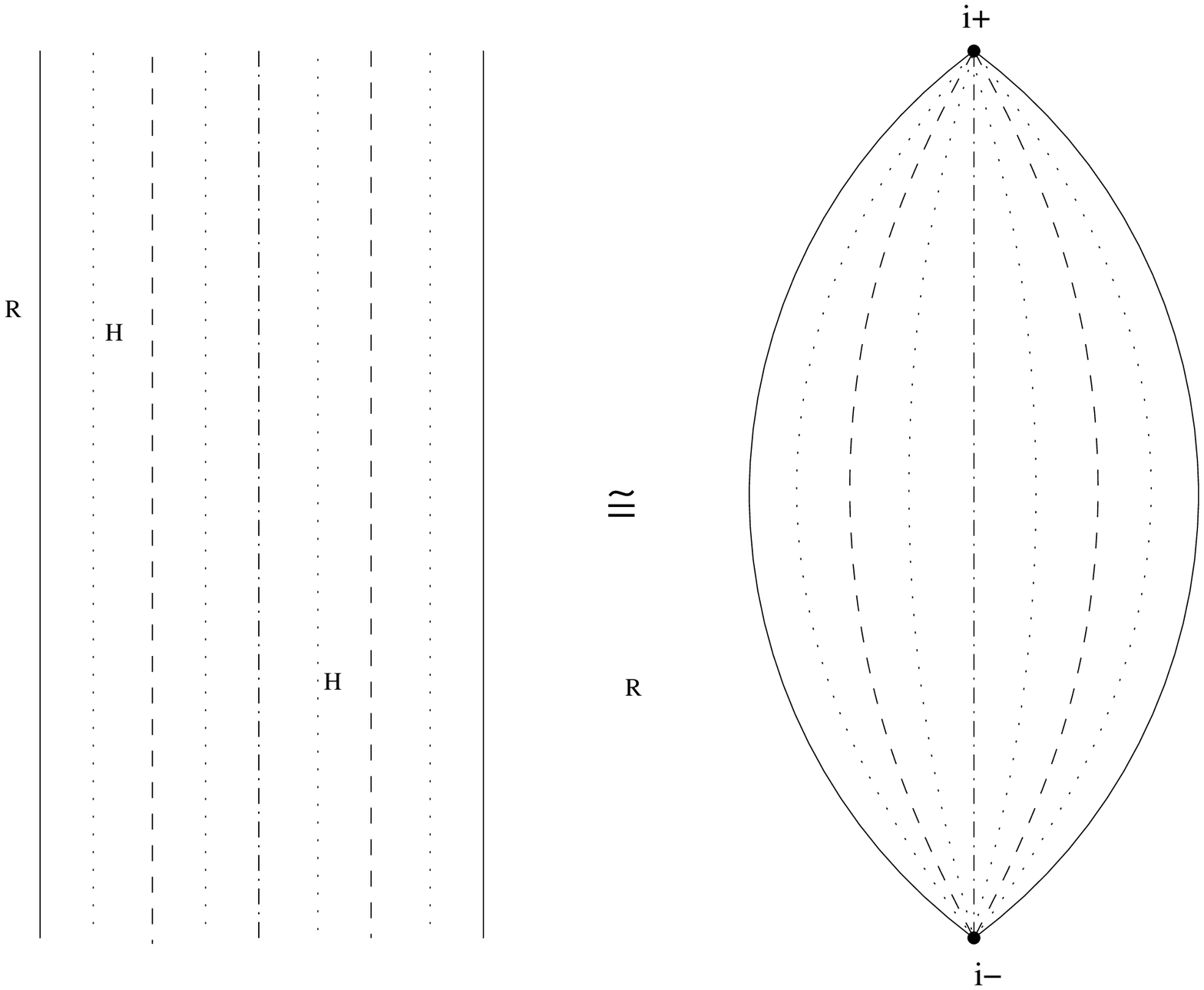}\label{Closed}}\qquad\qquad\qquad
 \subfigure[Bounded conformal diagram $k=0$ and $k=-1$. The dotted lines are curves of constant $R$ and the dashed line is the 
 apparent horizon $R=R_{AH}$. The thick solid line corresponds to a Big-Bang type (null) singularity.]
 {\includegraphics[width=0.22\textwidth]{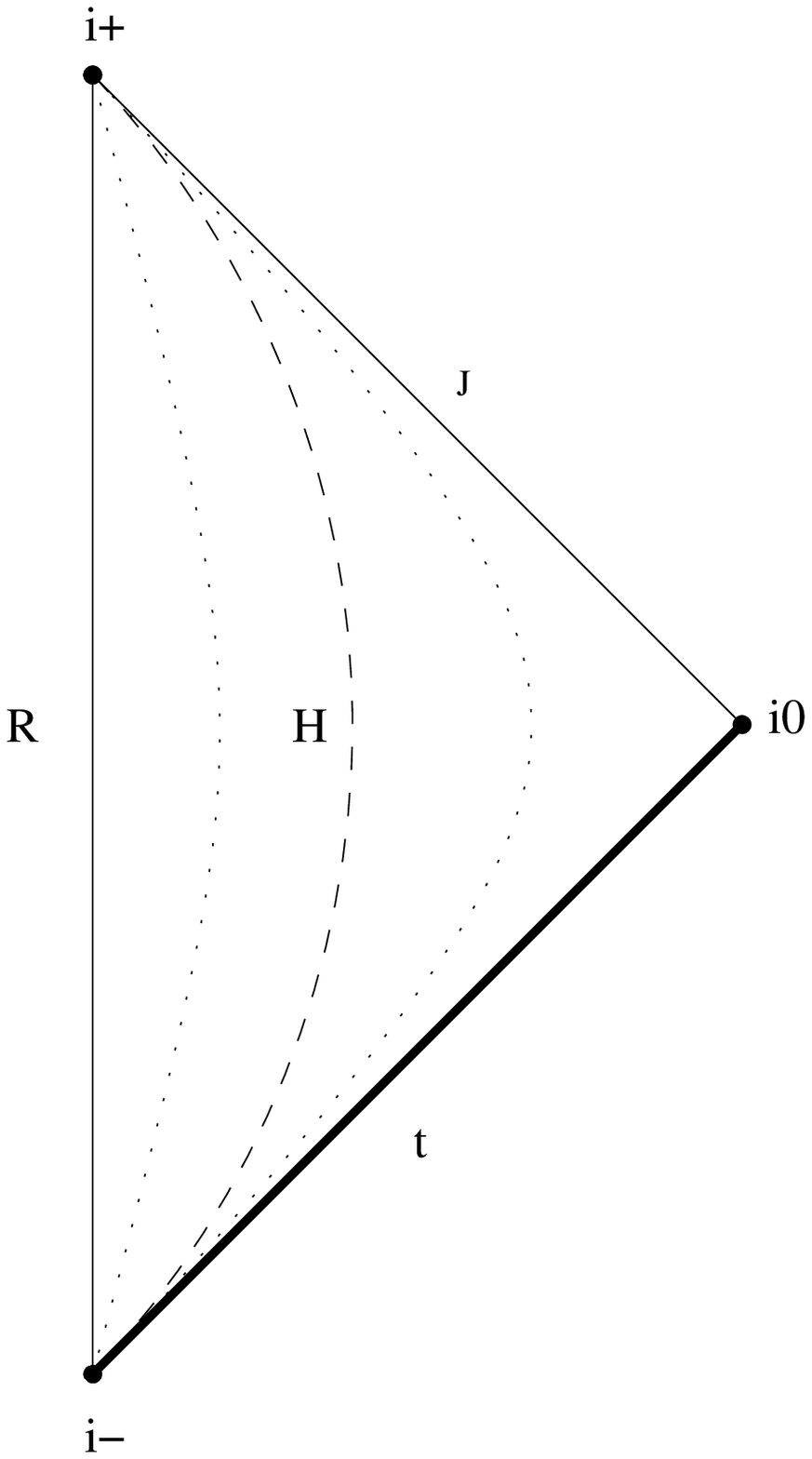}\label{RW_OpenFlat}}
\end{center}
\caption{Penrose diagrams for the expanding (at constant rate) interior solution. 
Each point represents a 2-sphere of radius $r=a(t)R$. As usual, $i^{-}$ and $i^{+}$ represent past and 
future timelike infinity respectively, and $i^{0}$ corresponds to spacelike infinity.
Also, $\mathscr{I^-}$, $\mathscr{I^+}$ denote past and future null infinity respectively.}\label{fig1}
\end{figure}

Note that we shifted the origin of time so that the singularity appears at $t=0$. 
This is a curvature singularity as the Ricci scalar curvature $\mathcal{R}_\mathrm{int}$ satisfies
\[
\mathcal{R}_\mathrm{int}=\frac{9\beta}{\delta_k}a(t)^{-2}\to +\infty \quad\text{as $a(t)\to 0^+$.}
\]
Toward the future the solution is forever expanding ($\dot a>0$) without acceleration ($\ddot a=0$). There is only one apparent horizon for this solution, located at
\begin{equation}\label{interiorAH}
R_\mathrm{AH}=\frac{1}{\sqrt{\delta_k^2+k}}.
\end{equation}
Let $R=R_b$ the boundary of the star. We distinguish three cases:
\begin{itemize}
\item[(i)]  $R_b>R_\mathrm{AH}$; the interior has an apparent horizon in this case.
\item[(ii)] $R_b=R_\mathrm{AH}$; the boundary of the star coincides with the apparent horizon.
\item[(iii)] $R_b<R_\mathrm{AH}$; the interior has no apparent horizon.
\end{itemize}
Now, let $r=r_\Sigma(u)$ the matching surface as seen by exterior observers. We assume that the exterior metric is given by~\eqref{vaidya2}. From~\eqref{rsigma},~\eqref{dtdw2} we obtain
\[
u=C_k t,\quad r_\Sigma(u)=x_\Sigma u,\quad x_\Sigma=\frac{R_b\delta_k}{C_k},\quad C_k=\frac{1}{\sqrt{1-kR_b^2}+\delta_kR_b} >0,
\] 
and~\eqref{vaidya2} becomes
\begin{equation}\label{vaidya3}
g_\mathrm{ext}=-\Big(1-\lambda_1\frac{u}{r}-\lambda_2\frac{r^2}{u^2}\Big)du^2-2dudr+r^2d\Omega^2,\quad r>x_\Sigma u,
\end{equation}
where
\[
\lambda_1=\frac{\beta R_b^3}{C_k},\quad \lambda_2=\frac{\beta C_k^2}{2\delta_k^3}.
\]
The mass and the cosmological scalar field in the exterior are given by 
\[
M(u)=\lambda_1u/2,\quad \Lambda(u)=3\lambda_2/u^2.
\] 
Note also that the metric~\eqref{vaidya3} has a curvature singularity at $u=0$, for its Kretschmann scalar $K=\mathrm{Riem}^2$ is given by
\[
K=\frac{12\lambda_1^2u^2}{r^6}+\frac{24\lambda_2^2}{u^4}.
\]

The apparent horizons of the metric~\eqref{vaidya3} are the hypersurfaces where $B(u,r)=0$, where
\begin{equation}\label{hdef}
B(u,r)=h\big(\frac{r}{u}\big),\quad h(x)=1-\frac{\lambda_1}{x}-\lambda_2 x^2.
\end{equation}
\begin{proposition}
The following holds:
\begin{itemize}
\item[(1)] In case (i) there is no apparent horizon in the exterior region and $B(u,r)<0$ for all $r>r_\Sigma(u), u>0$.
\item[(2)] In case (ii) the apparent horizon of the metric~\eqref{vaidya3} coincides with the apparent horizon of the interior, as well as with the matching surface:
\[
r_\mathrm{AH}(u)=r_\Sigma(u)\equiv R_\mathrm{AH}=R_b.
\]
Moreover $B(u,r)<0$ for all $r>r_\Sigma(u)$.
\item[(3)] In case (iii) there exists $x_\mathrm{AH}>x_\Sigma$ such that the metric~\eqref{vaidya} has an apparent horizon at $r=x_\mathrm{AH}u$. Moreover $B(u,r)>0$ for $r_\Sigma(u)<r<x_\mathrm{AH}u$ and $B(u,r)<0$ for $r>x_\mathrm{AH}u$.
\end{itemize}
\end{proposition}
\begin{proof}
For the proof it suffices to notice that the function $h(x)$ in~\eqref{hdef}
attains its maximum at $x=x_\Sigma$ and $h(x_\Sigma)=1-(R_b/R_\mathrm{AH})^2$. 
\end{proof}
The model under discussion is self-similar, with the lines $r=xu$, $x>x_\Sigma$, on the $(t,r)$-plane being tangent to the homothetic vector field in the exterior. We call such curves homothetic curves.
\begin{theorem}
There exists $x_*>x_\Sigma$ such that the homothetic curve $r=xu$ is spacelike for $x>x_*$, null for $x=x_*$ and timelike for $x_\Sigma<x<x_*$. In case (iii) there holds $x_*>x_\mathrm{AH}$. Moreover the homothetic curve $r=x_* u$ is the first ingoing radial null geodesics that escapes to null-infinity. 
\end{theorem} 
\begin{proof}
The metric induced on the hypersurface $r=xu$ in the exterior is given by
\[
-f(x)\,du^2+r^2d\Omega^2,\quad f(x)=h(x)+2x,
\]
where $h$ is the function defined in~\eqref{hdef}. The hypersurface $r=xu$ is spacelike if $f(x)<0$, timelike if $f(x)>0$ and null if $f(x)=0$.   A straightforward calculation shows that
\[
f(x_\Sigma)=\frac{1}{C_k^2}>0.
\]
As $f(x)\to -\infty$ when $x\to\infty$, there exists one (and only one) $x_*>0$ such that $f(x)>0$ for $x\in (x_\Sigma,x_*)$, $f(x_*)=0$ and $f(x)<0$ for $x>x_*$. The first part of the result follows. Moreover $f(x_\mathrm{AH})>0$ and so when the exterior apparent horizon exists we have $x_\mathrm{AH}<x_*$. As to the last statement, the equation for the ingoing radial null geodesics is 
\[
\frac{dr}{du}=-\frac{1}{2}B(u,r(u)).
\] 
Consider a solution $r(u)$ such that $r(u_0)/u_0<x_*$, for some $u_0>0$. Let $x(u)=r(u)/u$. As 
\[
\dot{x}(u)=-\frac{1}{2u}f(x(u)),
\] 
then $x(u)$ is decreasing for $x(u)<x_*$. It follows that $x(u)<x(u_0)<x_*$, for all $u>0$. Hence $f(x(u))\geq C>0$, for all $u>0$ and therefore
\[
x(u)\leq x_*-\frac{C}{2}\log\frac{u}{u_0}.
\]
We conclude that for $u$ large enough the ingoing light ray $r(u)$ must hit the boundary $x(u)=x_\Sigma$ in cases (i)-(ii) and the apparent horizon $x(u)=x_\mathrm{AH}$ in case (iii). In the latter case the ingoing light ray will enter the regular region and therefore hit the surface $x(u)=x_\Sigma$ of the star in finite time. In conclusion, all ingoing light rays escaping to infinity other than the null homothetic curve must satisfy $x(u)>x_*$, for all $u>0$, which concludes the proof of the theorem.  
\end{proof}
The result of the previous theorem leads to the Penrose diagram shown in Figure~\ref{fig2}. 

\begin{figure}[ht!]
\begin{center}
\psfrag{i-}[cc][cc][0.8][0]{$i^{-}$}
\psfrag{i0}[cc][cc][0.8][0]{$i^{0}$}
\psfrag{i+}[ll][cc][0.8][0]{$i^{+}$}
\psfrag{R}[rc][cc][0.8][0]{$R=0$}
\psfrag{H}[ll][cc][0.8][0]{$x_{AH}$}
\psfrag{x}[ll][cc][0.8][0]{$x_{*}$}
\psfrag{t}[ll][bb][0.8][0]{$\mathscr{I^-}(u=0,r=0)$}
\psfrag{S+}[cc][cc][0.8][0]{$\mathscr{I^+}(r=\infty)$}
\psfrag{r}[cc][cc][0.8][0]{$u=0$, $r>0$}
 {\includegraphics[width=0.65\textwidth]{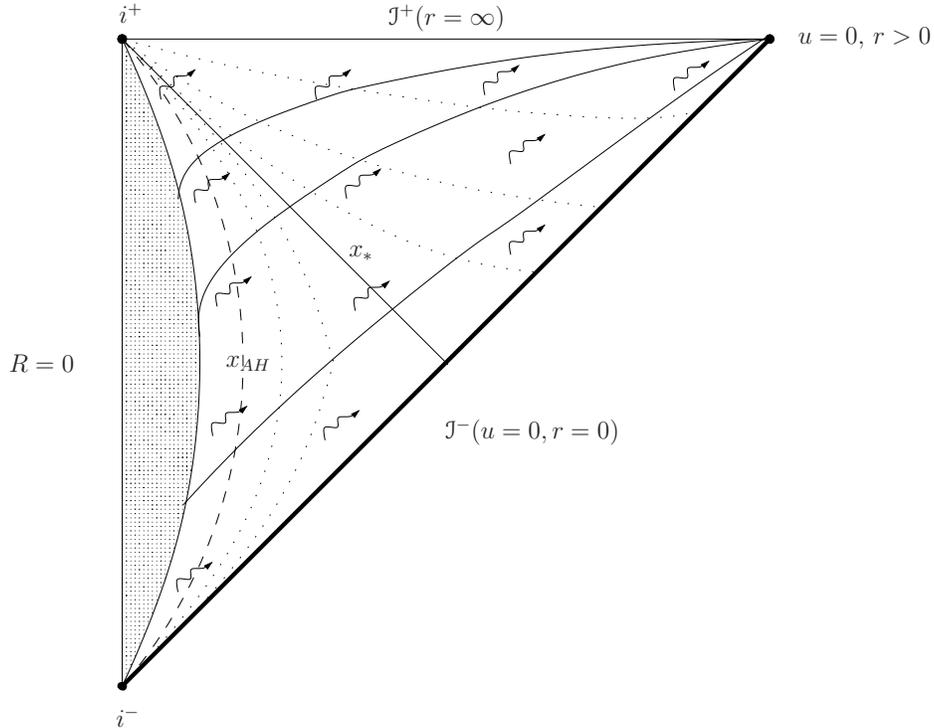}}
\end{center}
\caption{Penrose diagram for the self-similar stellar solution with diffusion in case (iii) 
showing the (homothetic) apparent horizon $x_{AH}$ in the exterior region (dashed timelike curve). 
The null homothetic (cosmological) horizon  $x_{*}$ is denoted by a solid line, while the other 
timelike and spacelike homothetic curves are denoted by dotted lines. 
The remaining solid lines are constant $r$ curves which are spacelike, becoming null at the apparent horizon, 
and are timelike afterwards. The thick solid line corresponds to the singularity. 
As usual, $i^{-}$ and $i^{+}$ represent past and future timelike infinity respectively, 
and $\mathscr{I^-}$, $\mathscr{I^+}$ denote past and 
future null infinity respectively.}\label{fig2}
\end{figure}

\section*{Acknowledgements}

A.~A. is supported by the projects EXCL/MAT-GEO/0222/2012, PTDC/MAT-ANA/1275/2014,
by CAMGSD, Instituto Superior T{\'e}cnico, by FCT/Portugal through
UID/MAT/04459/2013, and by the FCT Grant No. SFRH/BPD/85194/2012. 

The authors would like to thank J.~Nat\'ario for numerous discussions and the anonymous referee for valuable comments that helped to improve the paper.
Furthermore, A.~A.  thanks the Department of Mathematics at Chalmers University, Sweden, for the very kind hospitality.


\begin{thebibliography}{10}

\bibitem{OS39}
J. Oppenheimer, Snyder : On Continued Gravitational Contraction. 
Physical Review {\bf 56}, 455--459 (1939)

\bibitem{Chr84}
D. Christodoulou: Violation of Cosmic Censorship in the
Gravitational Collapse of a Dust Cloud . Commun. Math. Phys. {\bf 93}, 171--195 (1984)

\bibitem{Chr94}
D. Christodoulou: Examples of naked singularity formation in the gravitational collapse of a scalar field,
Ann. of Math. {\bf 140}, 607--653, (1994)

\bibitem{Chr99}
D. Christodoulou: The instability of naked singularities in the gravitational collapse of a scalar field.
Ann. of Math. {\bf 149}, 183-217 (1999) 

\bibitem{And} H. Andr\'easson: The Einstein-Vlasov System/Kinetic Theory. 
Living Rev. Relativ. (2011) 14: 4


\bibitem{ES79}
D. M. Eardley, L. Smarr. Time functions in numerical relativity: marginally bound dust collapse, 
Phys. Rev. D {\bf 19}, 2239 (1979)


\bibitem{Chr08}
D. Christodoulou: The Formation of Black Holes in General Relativity. 
European Mathematical Society Monographs in Mathematics, (2009)

\bibitem{calogero1}
S. Calogero: A kinetic theory of diffusion in general relativity with cosmological scalar field. JCAP 11/2011, 016 (2011)

\bibitem{calogero2}
S. Calogero: Cosmological models with fluid matter undergoing velocity diffusion. J. Geom. Phys.
{\bf 62}, 2208--2213 (2012)

\bibitem{calogero3} S.~Calogero, H.~Vetten: Cosmology with matter diffusion. JCAP 11/2013, 025 (2013)


\bibitem{FST96} 
F. Fayos, J. M. M. Senovilla, and R. Torres:
General matching of two spherically symmetric spacetimes.
Phys. Rev. D {\bf 54}, 4862 (1996)

\bibitem{GP09}
J. Griffiths and J. Podolsky.
Exact space-times in Einstein's General Relativity.
Cambridge Monographs on Mathematical Physics, (2009).


\bibitem{AC15} 
A. Alho, S. Calogero, M. P. Ramos, A. J. Soares: Dynamics of Robertson-Walker spacetimes with diffusion. 
Ann. Physics {\bf 354},  475--488 (2015)

\bibitem{nakao}
K.~Nakao: The Oppenheimer-Snyder space-time with a cosmological constant. Gen. Rel. Grav. {24}, 1069-1081 (1992)

\bibitem{dudley} R.~M.~Dudley: Lorentz-invariant Markov processes in relativistic phase space. Ark.~Mat. {\bf 6}, 241-–268 (1966)

\bibitem{ES45}
A. Einstein, E. G. Strauss, 
The influence of the expansion of space on the gravitation fields of surrounding the individual stars. 
Rev. Mod. Phys. 17, 120-124 (1945)

\bibitem{FJLS} F.~Fayos, X.~Jaen, E.~Llanta, J.~M.~M.~Senovilla: Matching the Vaidya and Robertson-Walker metric. Class.~Quant.~Gra. {\bf 8}, 2057-2068 (1991)


\bibitem{BOS89} 
W. B. Bonnor, A. K. G. de Oliveira, N. O. Santos: 
Radiating Spherical Collapse. 
Physics Reports {\bf 181} 5, 269--326 (1989)

















%
%









%


 
\end{thebibliography}
\end{document}